\documentclass[11pt,twoside]{amsart}

    \usepackage[english]{babel}
    \usepackage[T1]{fontenc}
    \usepackage[utf8]{inputenc}
    \usepackage{amsmath,amsfonts,amssymb,amsthm}
    \usepackage{xspace}
    \usepackage[pdftex,bookmarks=true,bookmarksnumbered=true,pdfpagemode=UseOutlines,pdfstartview=FitH,%
        pdfpagelayout=OneColumn,colorlinks=true,urlcolor=blue,pdfborder={0 0 0},breaklinks]{hyperref}
    \usepackage[algo2e,linesnumbered,commentsnumbered,vlined,ruled,norelsize]{algorithm2e} 
    \usepackage{graphicx}
    \usepackage[osf,noBBpl]{mathpazo}
    \usepackage{microtype}

    \theoremstyle{plain}
    \newtheorem{theorem}{Theorem}[section]
    \newtheorem*{conjecture}{Conjecture}
    \newtheorem{lemma}[theorem]{Lemma}
    \newtheorem{corollary}[theorem]{Corollary}
    \newtheorem{proposition}[theorem]{Proposition}
    \theoremstyle{definition}
    \newtheorem{definition}[theorem]{Definition}
    
    \newcommand\FF{\ensuremath{\mathbb F}}
    \DeclareMathOperator\lin{Lin}
    \newcommand\ith{-th\xspace}
    \newcommand\I{\mathcal{I}}
    \newcommand\Ring{\mathcal{R}}
    \newcommand\defeq{\stackrel{\text{\tiny def}}{=}}
    \renewcommand\L{\mathcal{L}}
    \newcommand\prog[1]{\textsc{#1}}
    \newcommand\Mult{\prog{Mult}}
    \newcommand\Clean{\prog{Clean}}
    \newcommand\Add{\prog{Add}}
    \newcommand\Iso{\prog{Isolate}}
    \newcommand\Prep{\prog{Preparation}}
    \newcommand\IsFactor{\prog{IsFactorizable}}
    \newcommand\SymDet{\prog{SymDet}}
    \newcommand\Merge{\prog{Merge}}

    \newcommand*\svg[4][htbp]{%
      \begin{figure}[#1]%
      \begin{center}%
      #4 
      \caption{#3}%
      \label{fig:#2}
      \end{center}%
      \end{figure}%
    }

\title[Symmetric Determinantal Representations in characteristic $2$]{Symmetric Determinantal Representations\\ in characteristic $2$}
\date{\today}

\author{Bruno Grenet}
\address{LIP, UMR 5668, \'ENS de Lyon -- CNRS -- UCBL -- INRIA\\%
    Universit\'e de Lyon}
\email{Bruno.Grenet@ens-lyon.fr}

\author{Thierry Monteil}
\address{LIRMM, UMR 5506, CNRS, Universit\'e Montpellier II}
\urladdr{http://www.lirmm.fr/~monteil/}

\author{St\'ephan Thomass\'e}
\address{LIP, UMR 5668, \'ENS de Lyon -- CNRS -- UCBL -- INRIA\\%
    Universit\'e de Lyon}
\email{Stephan.Thomasse@ens-lyon.fr}

\begin{document}

\begin{abstract}
This paper studies Symmetric Determinantal Representations (SDR) in
characteristic $2$, that is the representation of a multivariate
polynomial $P$ by a symmetric matrix $M$ such that $P=\det(M)$, and where
each entry of $M$ is either a constant or a variable. 

We first give some sufficient conditions for a polynomial to have an SDR.
We then give a non-trivial necessary condition,
which implies that some polynomials have no SDR, answering a question of
Grenet et al. 

A large part of the paper is then devoted to the case of multilinear
polynomials. We prove that the existence of an SDR for a multilinear
polynomial is equivalent to the existence of a factorization of the
polynomial in certain quotient rings. We develop some algorithms to test
the factorizability in these rings and use them to find SDRs when they
exist. Altogether, this gives us polynomial-time algorithms to factorize
the polynomials in the quotient rings and to build SDRs.
We conclude by describing the case of Alternating Determinantal
Representations in any characteristic.
\end{abstract}


\maketitle

\section{Introduction} 

Let $\FF$ be some field of characteristic $2$. 
A Symmetric Determinantal Representation (SDR) of a polynomial $P\in\FF[x_1,\dots,x_m]$ is a symmetric matrix $M$ with entries 
in $\FF\cup\{x_1,\dots,x_m\}$ such that $\det(M)=P$. One can also find in the literature other definitions where for instance 
the symmetric matrix has linear (degree-$1$) polynomials as entries. 
The two definitions are essentially equivalent, and we shall see that for our purposes, taking one or the other does not make any difference.

Symmetric Determinantal Representations have been studied at least from the beginning of the twentieth century~\cite{Dix02,Dic21} and apparently even from the nineteenth century~\cite{Bea00}. Definite SDRs are SDRs with the additional requirement that the matrix obtained by setting all the variables to zero is positive semi-definite. Definite SDRs play an important role in convex optimization, leading to a renew of interest in these representations, definite or not, in the recent years~\cite{HMV06,HV07,Bra10,NT12,GKKP11,PSV11,NPT11,Qua12}, see also \cite{Bea00} and the presentation~\cite{Sch09} for more perspectives on this. Recently, Petter Br\"and\'en has given SDRs for the elementary symmetric polynomials~\cite{Bra12}. He uses at this end graph-theoretic constructions and considers the laplacian matrix of the graph. Our constructions are also graph-theoretic but we consider the adjacency matrix of the graphs we obtain. Independently, symmetric determinants in characteristic two have also been a subject of studies~\cite{Alb38,Wat85}.

Symmetric Determinantal Representations for polynomials represented by weakly-skew circuits were given in \cite{GKKP11} for any field of characteristic different from $2$. The authors conjectured that these representations do not always exist in characteristic $2$. We prove this fact in this paper. 
To this end, we give a necessary condition for a polynomial to admit an SDR. We then focus on multilinear polynomials. For these polynomials, we show an equivalence between the existence of an SDR and the ability to factorize the polynomial in certain quotient rings. We develop algorithms to study the factorization in these quotient rings. Altogether, we obtain polynomial-time algorithms to factorize polynomials in the quotient rings and to compute SDRs of multilinear polynomials when they exist.

\begin{definition}
A polynomial $P\in\FF[x_1,\dots,x_m]$ is said \emph{representable} if it has an SDR, that is if there exists a symmetric matrix $M$ with entries 
in $\FF\cup\{x_1,\dots,x_m\}$ such that $P=\det(M)$. In this case, we say that $M$ \emph{represents} $P$.
\end{definition}

For instance, the polynomial $xy+yz+zx$ is representable as 
the determinant of the $4\times 4$ matrix 
    \[\begin{pmatrix}x & 0 & 0 & 1 \\ 0 & y & 0 & 1 \\ 0 & 0 & z & 1 \\ 1 & 1 & 1 & 0\end{pmatrix}.\]

Note that we ask the matrix to have entries 
in $\FF$. A natural relaxation would be to allow entries 
in an extension $\mathbb G$ of $\FF$. Actually, we shall show along the way that at least for multilinear polynomials, and most certainly 
for any polynomial, this relaxation is irrelevant. In the case of multilinear polynomials, Corollary~\ref{cor:baseField} shows that if a polynomial is representable, it has an SDR which only uses elements from the field generated by its 
coefficients.

\vspace{1em}
\textbf{Organization.} 
We begin by introducing some relevant algebraic background in
Section~\ref{sec:background}.
Section~\ref{sec:representable} is devoted to prove that SDRs exist for a
large class of polynomials. Then Section~\ref{sec:obstructions} proves the
main results of this paper: Some polynomials are not representable, and we
can characterize the multilinear representable polynomials. Some partial
results towards a full characterization are also given.
Section~\ref{sec:factor} is devoted to more algorithmic results. 
Using the equivalence between representability and factorizability in
certain quotient rings, we develop algorithms for these two tasks.
Section~\ref{sec:alternating} is devoted to the case of Alternating
Determinantal Representations in any characteristic.
Finally, we conclude in Section~\ref{sec:conclusion} by some remaining
open questions.

\vspace{1em}
Experimentations were done using the free open-source mathematics software
system \emph{Sage} \cite{Sage}, they allowed in return to fix a bug in its
determinant method (ticket \#10063). The algorithms presented in this paper
have been implemented and are available at 
\url{http://perso.ens-lyon.fr/bruno.grenet/publis/SymDetReprChar2.sage}.

\section{Algebraic background}\label{sec:background} 
Let us introduce some useful notions and notations. 

\subsection{Polynomials and determinants in characteristic $2$}
Let $\FF$ be any field of characteristic $2$, and let $\FF[x_1,\dots,x_m]$
be the ring of polynomials in $m$ indeterminates over $\FF$. 

Let $\alpha=(\alpha_1,\dots,\alpha_m)\in \mathbb{N}^m$, then the
\emph{primitive monomial} $x^\alpha$ is defined by
$x^\alpha=x_1^{\alpha_1}\dotsb x_m^{\alpha_m}$. 
A \emph{monomial} is a polynomial of the form $c x^\alpha$ for some
$c\in\FF$ and some primitive monomial $x^\alpha$. The constant $c$ is its
\emph{coefficient}, and $x^\alpha$ is its \emph{primitive part}. The value
$\deg x^\alpha=\sum_i\alpha_i$ is its \emph{total degree} and $\deg_i
x^\alpha=\alpha_i$ is its degree \emph{with respect to the variable
$x_i$}.
A polynomial is said to be \emph{multilinear} if its monomials
$c x_1^{\alpha_1}\dotsb x_m^{\alpha_m}$ satisfy $\alpha_i\leq 1$ for all
$i\leq m$.
\newline

Working in characteristic $2$ causes some inconveniences, like the
impossibility to halve.
But, it also simplifies some computations.
First, Frobenius endomorphism ensures that for any polynomials $P_1$ and
$P_2$, we have $(P_1+P_2)^2 = P_1^2 + P_2^2$.
Second, the determinant can easily be computed:
\begin{proposition}\label{prop:detchar2}
The determinant of an $(n\times n)$ symmetric matrix with entries in $\FF[x_1,\dotsc,x_m]$ is
\[\det(M)=\sum_\sigma \prod_{i=1}^n
M_{i,\sigma(i)},\] where $\sigma$ ranges over all involutions from
$\{1,\dots,n\}$ to itself, that is permutations such that
$\sigma^{-1}=\sigma$.
\end{proposition}
\begin{proof} The definition of the determinant is \[\det(M)=\sum_\sigma
\operatorname{sgn}(\sigma) \prod_{i=1}^n M_{i,\sigma(i)},\] where $\sigma$
ranges over all permutations of $\{1,\dots,n\}$. 
Actually, the signature of a permutation is either $1$ or $-1$, and those
two elements coincide in characteristic $2$. 
This means that the signature can be removed from the definition.
\newline
Consider $P_\sigma=\prod_i M_{i,\sigma(i)}$ for some permutation such that
$\sigma\neq\sigma^{-1}$. 
Then, $M_{i,\sigma^{-1}(i)}=M_{\sigma^{-1}(i),i}$ as $M$ is symmetric, and
$P_{\sigma^{-1}}=P_\sigma$. 
Thus the products for a permutation and its inverse cancel out in the sum. 
This shows that the sum can be restricted to involutions.
\end{proof}

\subsection{Quotient rings}\label{sec:quotientRings}

Given some polynomials $p_1,\dots,p_k$, we denote by 
$\langle p_1,\dots,p_k\rangle$ the ideal they generate. 
That is, \[\langle p_1,\dots,p_k\rangle = \left\{\sum_{i=1}^k p_iq_i :
q_i\in\FF[x_1,\dots,x_m]\right\}.\]
Given a tuple $\ell=(\ell_1,\dots,\ell_m)\in\FF^m$, we define the ideal
\[\I(\ell)=\langle x_1^2+\ell_1,\dots,x_m^2+\ell_m\rangle.\]
We also define the quotient ring $\Ring(\ell)$ as
$\FF[x_1,\dots,x_m]/\I(\ell)$ and denote by $\pi$ or $\pi_\ell$ the
canonical projection $\FF[x_1,\dots,x_m]\to\Ring(\ell)$.
The restriction of this projection to $\FF$ is one-to-one, hence $\FF$
naturally embeds into $\Ring(\ell)$, and the elements of $\FF\subseteq
\Ring(\ell)$ are called \emph{constants}.
This morphism of rings can be extended to matrices by $\pi(A)_{i,j} =
\pi(A_{i,j})$, and commutes with the determinant: 
$\pi \circ \det = \det \circ \pi$.
An element of $\Ring(\ell)$ is said to be \emph{linear} if it is the
projection of a linear polynomial.

Since the quotient identifies the squares of variables with constants, any
element of $r\in\Ring(\ell)$ has a unique multilinear representative in
$P\in\FF[x_1,\dots,x_m]$: we denote it by $\rho(r)$ or $\rho_\ell(r)$. We
have $\pi\circ\rho=\operatorname{Id}_{\Ring(\ell)}$. 
We denote by $\Mult_\ell$ or $\Mult$
the map $\rho_\ell\circ\pi_\ell$ that sends a polynomial to the
multilinear polynomial obtained by replacing each factor $x_i^2$ by
$\ell_i$.
For instance, let $P(x,y,z)=x^2y+z^3+xz+y$ then 
$\Mult_{(0,0,0)}(P)=xz+y$ and $\Mult_{(1,1,1)}(P)=y+z+xz+y=z+xz$.

The square of any element of $\Ring(\ell)$ belongs to $\FF$.
In particular, an element of $\Ring(\ell)$ is invertible if, and only if,
its square is not zero.
For example, $\pi(x_1x_2+1)$ is invertible if, and only if,
$\ell_1\ell_2\neq 1$.

Given a tuple $\ell=(\ell_1,\dots,\ell_m)\in\FF^m$, we denote by $\ell^2$
the tuple $(\ell_1^2,\dots,\ell_m^2)$ and say that $\ell^2$ is a
\emph{tuple of squares}.
If $\ell^2$ is a tuple of squares, the square of an element $r$ of
$\Ring(\ell^2)$ is the square of a unique element $c$ of $\FF$: we denote it
by $|r|$ or $|r|_{\ell^2}$, and call it the \emph{absolute value} of $r$. 
We remark that $|r_1r_2|=|r_1|\times|r_2|$ and 
$|r_1+r_2|=|r_1|+|r_2|$ for all $r_1,r_2\in\Ring(\ell^2)$. Furthermore,
$r$ is invertible if, and only if, $|r|\neq 0$.

\section{Some representable polynomials} \label{sec:representable} 

We deal with some positive results. Even though the main part of this paper is focused on negative results, we need to be able to represent some class of polynomials in order to give a characterization. 

In order to clarify some proofs, we will use the 
correspondence between permanents and cycle covers in graphs. We refer the reader to~\cite{Diestel} for the definitions concerning graphs.
Let $G$ be a weighted digraph and $M$ its adjacency matrix. We assume 
that the weights of $G$ are elements of $\FF[x_1,\dots,x_m]$. 
A \emph{cycle} is a set of distinct arcs $\{(v_1,v_2),(v_2,v_3),\dotsc,(v_{k-1},v_k)\}$ such that all the $v_i$'s are distinct but $v_1=v_k$. 
A \emph{cycle cover} of $G$ is a set of disjoint cycles such that each vertex of the digraph belongs to exactly one cycle. The \emph{weight} of a cycle cover is the product of the weights of all the arcs it uses. It is easily seen from the definition 
that the permanent of $M$ equals the sum of the weights of all the cycle covers of $G$. Since the characteristic of $\FF$ is two, the permanent of $M$ equals its determinant.

Suppose now that $G$ is symmetric (that is $M$ is symmetric). Proposition
\ref{prop:detchar2} shows that only some special cases of cycle covers can
be considered. More precisely, the determinant of $M$ equals the sum of
the weights of the cycle covers of $G$ corresponding to an involution.
These cycle covers are made of length-$1$ and length-$2$ cycles, and are
called \emph{partial matchings}. 

As $G$ is symmetric, it can actually be considered as an undirected graph. Length-$1$ cycles are \emph{loops}, and length-$2$ cycles are \emph{edges}. The weight of a length-$2$ cycle is the product of the weights of its 
arcs, that is the square of the weight of the edge. Thus consider a partial matching of a graph $G$ with (symmetric) adjacency matrix $M$. It can be viewed as a set $\mu$ of edges such that no vertex belongs to two 
distinct 
edges.  The discussion is summarized by the identity 
\[\det(M)=\sum_\mu \Bigl(\prod_{e\in\mu} w(e)^2 \times \prod_{v\notin\mu} w(v)\Bigr),\]
where $w(e)$ and $w(v)$ represent the weights of an edge $e$ and of a loop on a vertex $v$ respectively, $v\notin\mu$ means that the vertex $v$ is not covered by $\mu$, and $\mu$ ranges over all partial matchings of $G$.
An example is given by Figure~\ref{fig:ex-graph}: The adjacency matrix of the graph is given is the introduction. The only partial matchings are made of one of the three edges, to cover the central vertex, and two loops. By convention, an edge with no indicated weight has weight $1$.

\svg{ex-graph}{Graph representing $xy+xz+yz$.}{%
  \setlength{\unitlength}{60bp}%
  \begin{picture}(1,0.86666667)%
    \put(0,0){\includegraphics[width=\unitlength]{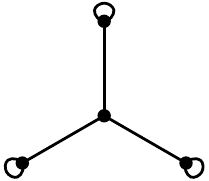}}%
    \put(-0.0260047,-0.03313245){\makebox(0,0)[b]{\smash{$y$}}}%
    \put(1.0372615,-0.04127047){\makebox(0,0)[b]{\smash{$z$}}}%
    \put(0.50039057,0.88657064){\makebox(0,0)[b]{\smash{$x$}}}%
  \end{picture}%
  }

In the following, if $M$ is a symmetric matrix, we denote by $G(M)$ the graph whose adjacency matrix 
is $M$. Conversely, given a graph $G$, we denote by $M(G)$ its adjacency matrix.
By a slight abuse of language, we shall say that a graph \emph{represents} a polynomial when its adjacency matrix is an SDR of the polynomial. In the same way, we write $\det(G)$ instead of $\det(M(G))$ to simplify the notations.
If $U=\{v_1,\dotsc,v_k\}$ is a subset of the vertices of a graph $G$, then $G\setminus U$ represents the induced subgraph of $G$ obtained by deleting the vertices of $U$ in $G$. For an edge $e$ of $G$, the graph $G-e$ is the graph obtained by deleting $e$ from $G$, but keeping its both extremities.

\begin{lemma}\label{lem:prod}
Let $P$ and $Q$ 
be two 
representable polynomials. Then $(P\times Q)$ is representable.
\end{lemma}

\begin{proof}
Let $M$ and $N$ be two symmetric matrices representing $P$ and $Q$ respectively. To represent the product by a graph, it is enough to consider the disjoint union of $G(M)$ and $G(N)$. This means that the SDR of $(P\times Q)$ is a block-diagonal matrix with two blocks being $M$ and $N$.
\end{proof}

The first part of the next lemma was proved in \cite{GKKP11}. We give here another proof which is suitable for the second part. 

\begin{lemma}\label{lem:squares}
Let $P\in\FF[x_1,\dots,x_m]$. Then $P^2$ is representable.

Moreover, there exists a graph $G$ that represents $P^2$ with two distinguished vertices $s$ and $t$ and such that $\det(G\setminus\{s,t\})=1$ and $\det(G\setminus\{s\})=\det(G\setminus\{t\})=0$.
\end{lemma}

\begin{proof}
Let $P=\sum_{\alpha\in\mathbb N^m} c_\alpha x^\alpha$ where $x^{\alpha}=x_1^{\alpha_1}\dotsm x_m^{\alpha_m}$. The square of a monomial $c_\alpha x^\alpha$ can be represented by a graph $G_\alpha$ of size $(2\deg(x^\alpha)+2)$. For a variable $x_i$ with exponent $\alpha_i$, we build $\alpha_i$ copies of a graph with two vertices and an edge of weight $x_i$ inbetween. We also build a graph with two vertices and an edge of weight $c_\alpha$ inbetween. These $(\deg (x^\alpha)+1)$ size-$2$ graphs are arranged in a line to build $G_\alpha$: The graphs are arranged in some arbitrary order and an edge of weight $1$ links two consecutive graphs (Figure~\ref{fig:chemin}).
\svg{chemin}{Graph $G_\alpha$ corresponding to some monomial $c_\alpha x^\alpha$ with $\alpha_1\ge 1$ and $\alpha_m\ge 2$.}{%
  \setlength{\unitlength}{154.51407471bp}%
  \begin{picture}(1,0.02529586)%
    \put(0,0){\includegraphics[width=\unitlength]{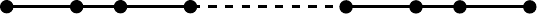}}%
    \put(0.08387192,0.04272353){\makebox(0,0)[b]{\smash{$c_\alpha$}}}%
    \put(0.92525382,0.04272353){\makebox(0,0)[b]{\smash{$x_m$}}}%
    \put(0.28938381,0.04080218){\makebox(0,0)[b]{\smash{$x_1$}}}%
    \put(-0.05545727,-0.00103428){\makebox(0,0)[b]{\smash{$s$}}}%
    \put(1.05847286,-0.00191051){\makebox(0,0)[b]{\smash{$t$}}}%
    \put(0.7096788,0.04272353){\makebox(0,0)[b]{\smash{$x_m$}}}%
  \end{picture}%
  }
The extremities of $G_\alpha$ are named $s$ and $t$. There is no loop in $G_\alpha$, therefore partial matchings are perfect matchings. The only perfect matching of $G_\alpha$ is made of all the edges of weight different from $1$. The weight of such a matching is $c_\alpha^2 (x^\alpha)^2$. The only matching of $G_\alpha\setminus\{s,t\}$ is made of the edges of weight $1$, and has weight $1$. Since $G_\alpha$ has an even number of vertices, $G_\alpha\setminus\{s\}$ and $G_\alpha\setminus\{t\}$ have no perfect matching.

Given a graph $G_\alpha$ for each monomial of $P$, the graph $G$ is the union of these graphs in which all the vertices with name $s$ on the one hand, and all vertices with name $t$ on the other hand, are merged. The perfect matchings of $G$ are then made of a perfect matching of some $G_\alpha$, and perfect matchings of weight $1$ of $G_\beta\setminus\{s,t\}$ for all $\beta\neq\alpha$. The sum of the weights of the matchings of $G$ is $\det(G)=\sum_\alpha c_\alpha^2(x^\alpha)^2=P^2$. Furthermore, the only perfect matching of $G\setminus\{s,t\}$ is made of perfect matchings of $G_\alpha\setminus\{s,t\}$ for all $\alpha$, thus $\det(G\setminus\{s,t\})=1$. By the same parity argument as before, $\det(G\setminus\{s\})=\det(G\setminus\{t\})=0$.
\end{proof}

This allows us to represent in a quite simple way a large class of polynomials.

\begin{proposition}\label{prop:representable}
Let $P(x_1,\dots,x_m)=L_1\times L_2\times \dotsb \times L_k$, where for $1\le i\le k$,
\[L_i(x_1,\dots,x_m)=P_{i0}^2 + x_1 P_{i1}^2 + \dotsb + x_m P_{im}^2\]
for some $P_{ij}\in\FF[x_1,\dots,x_m]$. 
Then $P$ is representable.
\end{proposition}

\begin{proof}
By Lemma~\ref{lem:prod}, it is sufficient to show how to represent each $L_i$. 
We first prove how to represent a polynomial of the form
\[L(x_1,\dots,x_m)=\lambda_0^2+x_1\lambda_1^2+\dotsb+x_m\lambda_m^2,\]
where the $\lambda_j$'s are constants from $\FF$. 

\svg{linGraph}{Graph representing $L=\lambda_0^2+x_1\lambda_1^2+\dotsb+x_m\lambda_m^2$.}{%
  \setlength{\unitlength}{121.6bp}%
  \begin{picture}(1,0.81578947)%
    \put(0,0){\includegraphics[width=\unitlength]{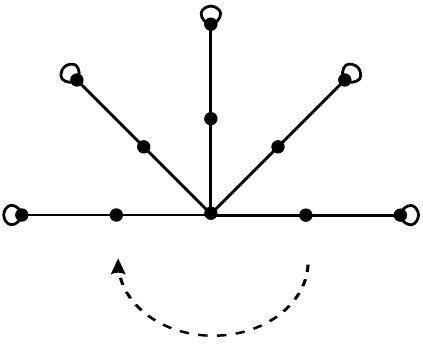}}%
    \put(0.5043773,0.81495814){\makebox(0,0)[b]{\smash{$1$}}}%
    \put(0.89618235,0.67439623){\makebox(0,0)[b]{\smash{$x_1$}}}%
    \put(-0.02038761,0.24336505){\makebox(0,0)[b]{\smash{$x_{m-1}$}}}%
    \put(1.05925323,0.2906023){\makebox(0,0)[b]{\smash{$x_2$}}}%
    \put(0.09509331,0.67917569){\makebox(0,0)[b]{\smash{$x_m$}}}%
    \put(0.55020256,0.44664731){\makebox(0,0)[b]{\smash{$\lambda_0$}}}%
    \put(0.67911895,0.37091422){\makebox(0,0)[b]{\smash{$\lambda_1$}}}%
    \put(0.61383141,0.23334264){\makebox(0,0)[b]{\smash{$\lambda_2$}}}%
    \put(0.40272114,0.23447863){\makebox(0,0)[b]{\smash{$\lambda_{m-1}$}}}%
    \put(0.36524817,0.36057162){\makebox(0,0)[b]{\smash{$\lambda_m$}}}%
  \end{picture}%
  }

The linear polynomial $L$ is represented by the graph $G_L$ given on Figure~\ref{fig:linGraph}. 
We 
prove that it effectively represents $L$: A partial matching has to match the central vertex with some of its neighbors. Once a neighbor is chosen, say in the direction of $x_i$, the loop with weight $x_i$ has to be chosen. Then, there is only one 
choice to have a partial matching which consists 
in covering the remaining vertices by the outside edges. Thus the weight of such a partial matching is $x_i\lambda_i^2$, and the sum over all partial matchings equals $L$.

Now, let $G_{P_i}$ be the graph representing the polynomial $P_i^2$ given by Lemma~\ref{lem:squares}, with its two distinguished vertices $s$ and $t$. By a slight abuse of language, we call $\lambda_i$ the edge that has weight $\lambda_i$ in $G_L$, and denote by $s_i$ and $t_i$ its 
extremities. Let $G_L-\lambda_i$ be the graph obtained from $G_L$ by removing the edge $\lambda_i$. We build a new graph $G_L'$ in which $G_{P_i}$ \emph{replaces} the edge $\lambda_i$: 
The graph $G_L'$ is the disjoint union of $G_L-\lambda_i$ and $G_{P_i}$, in which $s_i$ and $s$ (respectively $t_i$ and $t$) are merged.

A partial matching of $G_L$ either is a partial matching of $G_L-\lambda_i$, or is made of $\lambda_i$ and a partial matching of $G_L\setminus\{s,t\}$. Thus $\det(G_L)=\det(G_L-\lambda_i)+\lambda_i^2\det(G_L\setminus\{s,t\})$. In $G_L'$, a partial matching can also be of two sorts: Either it is made of partial matchings of $G_{P_i}\setminus\{s,t\}$ and $G_L-\lambda_i$, or of partial matchings of $G_{P_i}$ and $G_L\setminus\{s,t\}$. Indeed, no partial matching exists covering $G_{P_i}\setminus\{s\}$ (respectively $G_{P_i}\setminus\{t\}$). Thus 
\begin{align*}
\det(G_L') &=\det(G_{P_i}\setminus\{s,t\})\times\det(G_L-\lambda_i) + \det(G_{P_i})\times\det(G_L\setminus\{s,t\})\\
           &=  1\times\det(G_L-\lambda_i) + P_i^2\times\det(G_L\setminus\{s,t\}).
\end{align*}
This shows that we can replace in $G_L$ each $\lambda_i$ by the graph $G_{P_i}$ to obtain an SDR of
$P_0^2+x_1 P_1^2+\dotsb+x_m P_m^2$.
\end{proof}

In particular, this theorem shows that if $\FF$ is a finite field of characteristic $2$, every linear polynomial is representable since every element in such a field is a quadratic residue. 

\begin{definition}\label{gSDR}
Let $P\in\FF[x_1,\dotsc,x_m]$. A \emph{generalized Symmetric Determinantal Representation} (gSDR) of $P$ is a symmetric matrix $M$ such that $\det(M)=P$ and whose entries are polynomials of $\FF[x_1,\dotsc,x_m]$ such that each diagonal entry is 
either a constant or
of the form $P_0^2+x_1P_1^2+\dotsb+x_mP_m^2$ where $P_1,\dotsc,P_m\in\FF[x_1,\dotsc,x_m]$.
\end{definition}

\begin{theorem}\label{thm:representable}
A polynomial $P\in\FF[x_1,\dotsc,x_m]$ is representable if, and only if, it admits a gSDR. 
\end{theorem}

\begin{proof}
An SDR is already a gSDR. 
We once again work with the graph representation instead of the matrix representation. Suppose we have a graph $G$ where the weights of the edges are any polynomials, and the weights of the loops are 
either constants or 
of the form $P_0^2+x_1P_1^2+\dotsb+x_mP_m^2$. We show how we can turn this graph into an SDR.

We use the same technique as in the proof of Proposition~\ref{prop:representable} to replace each edge with weight $P$ in $G$ by the graph $G_P$ which is an SDR of $P^2$. It remains to show how to deal with the loops.

Suppose some vertex $v$ of $G$ has a loop of weight $L=P_0^2+x_1P_1^2+\dotsb+x_mP_m^2$. Consider the graph $G_L$ obtained in Proposition~\ref{prop:representable}, and let $G_0$ be the graph obtained from $G$ by removing the loop on $v$. Then $G$ is replaced by $G_0\cup G_L$, where the central vertex of $G_L$ is merged with $v$. Let $G'$ be this new graph. Note that $\det(G_L\setminus\{v\})=1$.
Then
\begin{align*}
\det(G')&=\det(G_L)\times \det(G\setminus\{v\})+\det(G_L\setminus\{v\})\times \det(G_0)\\
        &= L\times \det(G\setminus\{v\})+ 1\times \det(G_0)=\det(G).
\end{align*}
Repeating this operation for all the loops of the graph yields the result.
\end{proof}

\section{Obstructions to SDR}\label{sec:obstructions} 

This section deals with negative results, showing that some polynomials
have no SDR.  Section~\ref{sec:nec-cdt} is devoted to a necessary
condition that holds for any polynomial. It is followed by a simple
example of a polynomial with no SDR. We prove in
Section~\ref{sec:multilin} that this necessary condition is actually a
characterization when applied to multilinear polynomials. Finally,
Section~\ref{sec:full} gives some partial results towards a full
characterization.

\subsection{A necessary condition} \label{sec:nec-cdt} 

We aim to prove in this section a necessary condition for a polynomial to be representable. We introduce a notion of factorization \emph{modulo} some ideal $\I(\ell)$ to express this condition. 

\begin{definition}
Let $P\in\FF[x_1,\dots,x_m]$. Then $P$ is said \emph{factorizable modulo
$\I(\ell)$} if there exist some linear elements $t_1,\dotsc,t_k$ of
$\Ring(\ell)$ such that \[ \pi_\ell(P) = t_1\times\dotsb\times t_k.\]
\end{definition}

This definition can be restated as follows. A polynomial $P$ is factorizable 
\emph{modulo} $\I(\ell)$ if there exists some linear polynomials $L_1,\dotsc,L_k$
of $\FF[x_1,\dotsc,x_m]$ such that $\pi_\ell(P)=\pi_\ell(L_1\dotsm L_k)$.

\begin{theorem}\label{thm:nc}
Let $P\in\FF[x_1,\dots,x_m]$ be a representable polynomial. Then for every
tuple of squares $\ell^2\in\FF^m$, $P$ is factorizable \emph{modulo}
$\I(\ell^2)$.
\end{theorem}

For instance, one can recall the representable polynomial $P(x,y,z)=xy+yz+xz$ 
from the introduction. Then $\pi_{(0,0,0)}(P)=\pi_{(0,0,0)}((x+y)(x+z))$ and
$\pi_{(1,1,1)}(P)=\pi_{(1,1,1)}(xyz(x+y+z))$.

The proof of this theorem is of algorithmic nature. We give an algorithm that
takes as inputs an SDR $M$ of some polynomial $P$ and a tuple
of squares $\ell^2$, and returns a factorization of $P$ \emph{modulo} $\I(\ell^2)$.
The general idea is to build the projection $A=\pi(M)$ of $M$ to get a representation
of $\pi(P)$, and then to perform row and column operations to \emph{isolate} some
diagonal entry $A_{i,i}$, that is to cancel out each entry $A_{i,j}$ for $j\neq i$,
keeping $A$ symmetric. We then show that $A_{i,i}$ is a linear element of 
$\Ring(\ell^2)$. Thus we can write $\pi(P)=A_{i,i}\det(A')$ where $A'$ is 
obtained from $A$ by removing its row and column of index $i$. By induction on
the dimensions of $A$, we can conclude that $\pi(P)$ can be factorized as a product
of linear elements.
In what follows, we prove some lemmas that
justify this approach.

Let us fix some tuple of squares $\ell^2$. In the next definition,
we extend the notion of gSDR, originally defined for polynomials, to
elements of $\Ring(\ell^2)$.
\begin{definition}\label{gSDR:R}
Let $r\in\Ring(\ell^2)$. A \emph{generalized Symmetric Determinantal
Representation} (gSDR) of $r$ is a symmetric matrix $A$ such that $A$ has
linear diagonal entries and $\det(A) = r$.
\end{definition}

In a gSDR for a polynomial, the diagonal entries are 
either constants or of the form 
$P_0^2+x_1P_1^2+\dotsb+x_mP_m^2$. The projection of such a polynomial
is a linear element of $\Ring(\ell^2)$. Indeed, for all $i$, 
$\pi(P_i^2)=\pi(P_i)^2$ belongs to $\FF$. Therefore, if we let
$\lambda_i=\pi(P_i)^2$ for all $i$, $\pi(P_0^2+x_1P_1^2+\dotsb+x_mP_m^2)$ is also
the projection of the linear polynomial $\lambda_0+x_1\lambda_1+\dotsb+x_m\lambda_m$.

The previous remark implies in particular the following lemma:

\begin{lemma}\label{nc:proj}
Let $M$ be a gSDR of some polynomial $P\in\FF[x_1,\dots,x_m]$. Then the
matrix $\pi(M)$ is a gSDR of $\pi(P)$.
\end{lemma}

Next lemma gives some structure to the gSDR of an element of $\Ring(\ell^2)$.

\begin{lemma}\label{nc:diag} 
Let $A$ be a gSDR of some $r\in\Ring(\ell^2)$. 
Then there exists a gSDR $B$ of $r$ whose non-diagonal entries are
constants. 
\end{lemma}

\begin{proof}
Suppose that $A_{i,j}=A_{j,i}=\pi(P)$, $i\neq j$, for some polynomial $P$. 
Since the determinant of $A$ equals $\sum_\sigma\prod_i A_{i,\sigma(i)}$
where $\sigma$ ranges over the involutions (by
Proposition~\ref{prop:detchar2}), if $A_{i,j}$ divides a monomial in 
$\det(A)$, then so does $A_{i,j}^2$. Thus, if $\pi(P)$ divides a monomial, 
so does $\pi(P)^2$. 
If we replace $A_{i,j}$ and $A_{j,i}$ by the absolute value $|\pi(P)|\in
\FF$, the determinant of $A$ is unchanged as $\pi(P)^2=|\pi(P)|^2$ by
definition. 
This proves the lemma, as $B$ can be obtained by replacing each
non-diagonal entry by its absolute value.
\end{proof}

We now define the main tools we use to prove the theorem. These are simple
algorithmic transformations that we apply 
on the symmetric matrix representing an element $r\in\Ring(\ell^2)$ such that
the determinant remains unchanged and the matrix becomes diagonal. All of
these depend on the tuple $\ell^2$, even though it is not explicitly given
as an argument to simplify the notations.

Let $\Clean$ be the algorithm that replaces each non diagonal entry
$A_{i,j}$ by its absolute value $|A_{i,j}|$ as in
Lemma~\ref{nc:diag}. We define two other algorithms, $\Add_{i,j,\alpha}$
(Algorithm~\ref{algo:add}) and $\Iso_i$ (Algorithm~\ref{algo:iso}).

\begin{algorithm2e}[htbp]
\DontPrintSemicolon
\caption{$\Add_{i,j,\alpha}(A)$}
\label{algo:add}
$n\gets$ dimension of $A$\;
\lFor{$k=1$ to $n$}{$A_{j,k}\gets A_{j,k} + \alpha A_{i,k}$} \tcp*{$R_j\gets R_j + \alpha R_i$}
\lFor{$k=1$ to $n$}{$A_{k,j}\gets A_{k,j} + \alpha A_{k,i}$} \tcp*{$C_j\gets C_j+\alpha C_i$}
\Return $\Clean(A)$\;
\end{algorithm2e}

\begin{algorithm2e}[htbp]
\DontPrintSemicolon
\caption{$\Iso_i(A)$}
\label{algo:iso}
$n\gets$ dimension of $A$\;
\For{$j=1$ to $n$}{
    \If{$j\neq i$}{
        $\alpha\gets A_{i,j}\times |A_{i,i}|^{-1}$\;
        $A\gets \Add_{i,j,\alpha}(A)$\;
    }
}
\Return $A$\;
\end{algorithm2e}

\begin{lemma}
Let $A$ be a gSDR of some $r\in\Ring(\ell^2)$. Then $\Add_{i,j,\alpha}(A)$
is a gSDR of $r$ whose non-diagonal entries are constants.
\end{lemma}

\begin{proof}
The algorithm adds $\alpha$ times the $i$\ith{} row to the $j$\ith{} one,
and then $\alpha$ times the $i$\ith{} column to the $j$\ith{} one.  These 
two operations do not change the determinant. Furthermore, only
entries of the $j$\ith{} row and column are changed. But with those two
operations, $A_{j,k}$ is replaced by $A_{j,k}+\alpha A_{i,k}$ while $A_{k,j}$
is replaced by $A_{k,j}+\alpha A_{k,i}$ for $j\neq k$. As initially
$A_{i,k}=A_{k,i}$ and $A_{j,k}=A_{k,j}$, $A$ remains symmetric. Furthermore,
$A_{j,j}$ is first replaced by $A_{j,j}+\alpha A_{i,j}$, and finally by
$(A_{j,j}+\alpha A_{i,j})+\alpha (A_{i,j}+\alpha A_{i,i})=A_{j,j}+\alpha^2
A_{i,i}$. Thus $A_{j,j}$ remains linear. This shows that $A$
remains a gSDR of $r$ after the first two operations. 
Eventually, $\Clean$ 
is applied to 
a gSDR and we obtain the second property.
\end{proof}

\begin{lemma}\label{nc:iso} 
Let $A$ be a gSDR of some $r\in\Ring(\ell^2)$. If there exists an index
$i$ such that $|A_{i,i}|\neq 0$, then
$A'=\Iso_i(A)$ is a gSDR of $r$. Furthermore $A'_{i,j}=A'_{j,i}=0$ for all
$j\neq i$.
\end{lemma}

\begin{proof}
The matrix $\Iso_i(A)$ is a gSDR of $r$ since $\Add_{i,j,\alpha}(A)$ is a
gSDR of $r$ (for all $j$ and $\alpha$). 
Now, 
let $\alpha = A_{i,j}\times |A_{i,i}|^{-1}$ for some $j$ such that
$A_{i,j}\neq 0$ and consider the action of $\Add_{i,j,\alpha}$ 
on the $i$\ith{} row of $A$. 
The only altered entry is $A_{i,j}$, when the $i$\ith{} column multiplied
by $\alpha=A_{i,j}\times|A_{i,i}|^{-1}$ is added to the $j$\ith{} one,
and then by $\Clean(A)$. 
So $A_{i,j}$ is replaced by $A_{i,j}(1+|A_{i,i}|^{-1}\times
A_{i,i})$. Since $|A_{i,i}|^2=A_{i,i}^2$ by definition, $A_{i,j}$ is
replaced by $0$ during $\Clean(A)$. The same is true on the $i$\ith{}
column. Thus, 
if $A'=\Iso_i(A)$, $A'_{i,i}$ is 
the only nonzero entry in the $i$\ith{} row and column of $A'$.
\end{proof}

We recall that an element of $\Ring(\ell^2)$ is invertible
if, and only if, its square is nonzero.

\begin{lemma}\label{nc:square}
Let $A$ be a gSDR of some $r\in\Ring(\ell^2)$ such that 
no diagonal entry is invertible.
If there exists a nonzero diagonal entry, say $A_{1,1}$, and a nonzero
entry 
$A_{1,j}$ for $j>1$, then one can build a new gSDR of the same dimensions
$\tilde A$, representing some $\tilde r\in\Ring(\ell^2)$ such that
$r=(A_{1,1}+1)\times \tilde r$, where moreover $\tilde A$ contains 
some invertible diagonal entries.
\end{lemma}

\begin{proof}
Let us write $A_{1,1}$ as $1+(A_{1,1}+1)$. Let 
\[B=\begin{pmatrix}
    1 & 1        & 0       & \dots & 0         \\
    1 & A_{1,1}+1& A_{1,2} & \dots & A_{1,n}   \\
    0 & A_{2,1}  &         &       &           \\
    \vdots&\vdots&         & A'    &           \\
    0 & A_{n,1}  &         &       &       
    \end{pmatrix}\]
where $A'$ is obtained from $A$ by removing its 
first row and column. Then 
$\det(B)=\det(A)$. Indeed, adding the first row of $B$ to 
the second one, and the first column to the second one yields
the matrix 
\[
    \begin{pmatrix}
    1 & 0       & 0       & \dots & 0         \\
    0 & A_{1,1} & A_{1,2} & \dots & A_{1,n}   \\
    0 & A_{2,1} &         &       &           \\
    \vdots&\vdots&        & A'    &           \\
    0 & A_{n,1} &         &       &       
    \end{pmatrix}\]
whose determinant equals $\det(A)$.

Now
\[\Iso_2(B)= 
    \begin{pmatrix}
    A_{1,1} & 0  & A_{1,2} & \dots & A_{1,n}   \\
    0 & A_{1,1}+1& 0       & \dots & 0         \\
    A_{2,1} & 0  &         &       &           \\
    \vdots&\vdots&         & A''   &           \\
    A_{n,1} & 0  &         &       &       
    \end{pmatrix}\]
still has the same determinant. For each $j>1$, $A_{j,j}$ is replaced by $A_{j,j}+A_{1,j}^2(A_{1,1}+1)$ in $A''$. Since no diagonal entry is invertible by hypothesis, $A_{j,j}^2=0$ for all $j$ 
and $(A_{1,1}+1)^2=1$. Thus $A''$ contains some diagonal entries whose square is nonzero, that is some invertible diagonal entries. 
Actually, this holds since we supposed that some $A_{1,j}$ is nonzero. 
Now, the determinant of this matrix equals
\[(A_{1,1}+1)\times\det \begin{pmatrix}
    A_{1,1} & A_{1,2} & \dots & A_{1,n} \\
    A_{2,1} &         &       &         \\
    \vdots  &         & A''   &         \\
    A_{n,1} &         &       &       
    \end{pmatrix}.\]
Therefore, $A$ can be replaced by this new matrix $\tilde A$, of the same
dimensions, with some invertible diagonal entries. 
Then
$\tilde A$ is a gSDR for some $\tilde r\in\Ring(\ell^2)$ such that
$r=(A_{1,1}+1)\times\tilde r$.
\end{proof}

We now have all the ingredients to prove the theorem.

\begin{proof}[Proof of Theorem~\ref{thm:nc}]
Let us first restate the theorem. We aim to prove that if
$P\in\FF[x_1,\dots,x_m]$ has a gSDR, then its projection $r = \pi(P)$ can
be written as $t_1\times\dotsb\times t_k$ where
$t_1$, \dots, $t_k$ are linear elements of $\Ring(\ell^2)$. 
Suppose we are given a gSDR $M$ of some polynomial $P$. Then we have a gSDR
$A=\pi(M)$ of $r=\pi(P)$ by Lemma~\ref{nc:proj}. 
Thus we have to prove that given a gSDR $A$ of some $r\in\Ring(\ell^2)$, we
can find some linear elements $t_1,\dots,t_k$ of $\Ring(\ell^2)$ such that
$r=t_1\times\dotsb\times t_k$. 

First note that if $A$ does not satisfy the conditions of
Lemma~\ref{nc:square}, then we can already conclude. Indeed, this means
that each diagonal entry is either zero, or is the only nonzero entry in
its row and column. By reordering the rows and columns, we can get a
block-diagonal matrix with two blocks: The first one has zero diagonal
entries and the second one is diagonal. Therefore, since the determinant
of $A$ is the product of the determinants of these two blocks, we get a
constant times a product of linear elements. In other words, the
factorization is found.

So let $A$ be a gSDR of some $r\in\Ring(\ell^2)$
satisfying the hypotheses of Lemma~\ref{nc:square}. 
We build a gSDR $\tilde A$ of some $\tilde r\in\Ring(\ell^2)$ such that
$r=t\times\tilde r$ for some linear element $t\in \Ring(\ell^2)$, and such
that $\tilde A$ has at least one invertible diagonal entry $\tilde A_{i,i}$.
If $A$ already satisfied the property, then $\tilde A=A$ and $t=1$. 
Now, the $i$\ith{} row and column of $A'=\Iso_i(\tilde A)$ have as only
nonzero entry $A'_{i,i}$ by Lemma~\ref{nc:iso}. Thus, removing the
$i$\ith{} row and column to $A'$ 
yields a gSDR $B$ of some element $s\in\Ring(\ell^2)$ such that $\tilde
r=A'_{i,i}\times s$. 

This shows that from a gSDR of dimensions $(n\times n)$ of some
$r\in\Ring(\ell^2)$, we can build a gSDR of dimensions $(n-1)\times (n-1)$
of some $s\in\Ring(\ell^2)$ such that $r=t\times t'\times s$ 
where $t$ and $t'$ are linear. 

We can now use 
induction to prove that if $r\in\Ring(\ell^2)$ has a gSDR, then it can be
written as $t_1\times\dotsb\times t_k$ for some linear elements
$t_1,\dotsc,t_k\in\Ring(\ell^2)$. Indeed, if $A$ is a $(1\times 1)$ gSDR
of $r$, then $r$ is linear. 
\end{proof}

The proof of Theorem~\ref{thm:nc} is of algorithmic nature. 
It is easily seen that the underlying algorithm runs in time polynomial in the dimensions of the input gSDR. (More precisely, the complexity of the algorithm is $\mathcal O(n^3)$.)

\subsection{An example}\label{sec:example} 
Let us consider the 
polynomials in $\FF_2[x,y,z]$, where $\FF_2$ denotes the field with two
elements.
The ring $\Ring(1,1,1)$ has 256 elements, 136 of which can be written as
the product of linear polynomials, 120 of which can not. The element
$\pi(xy+z)$ is one of those. Therefore, Theorem \ref{thm:nc} tells us that
the polynomial $xy+z$ can \emph{not} be represented as the determinant of
a symmetric matrix with entries in $\FF_2 \cup \{x,y,z\}$.

\subsection{Multilinear polynomials}\label{sec:multilin} 

In this section, we show that the necessary condition of
Theorem~\ref{thm:nc} is actually a characterization when applied to
multilinear polynomials. This relies on the following structural lemma. It
is valid for any polynomial, even non-multilinear.

\begin{lemma}\label{lem:multilin}
Let $P\in\FF[x_1,\dots,x_m]$ be a representable polynomial. Then there
exists an SDR $M$ of $P$ such that each variable appears at most once on
the diagonal. 
\end{lemma}

\begin{proof}
Let $M$ be any SDR for $P$, that is $\det(M)=P$, and $M$ has entries 
in $\FF\cup\{x_1,\dots,x_m\}$. 
Suppose 
that for some other $i$, $x_i$ appears (at least) twice on the diagonal, as entries
$M_{i_1,i_1}$ and $M_{i_2,i_2}$. 
Consider the matrix obtained after adding the row 
of index $i_1$ to the row 
of index $i_2$,
and the column of index $i_1$ to the column of index $i_2$. 
As already mentioned, the only altered diagonal entry
is $M_{i_2,i_2}$ and it is now equal to $M_{i_2,i_2}+M_{i_1,i_1}=2x_i=0$. Therefore, we obtain a new SDR with one occurrence of $x_i$ on the diagonal replaced by zero. We can repeat this for each variable until each variable appears at most once on the diagonal.
\end{proof}

We can use this lemma to obtain the desired characterization, when $\FF$ is a finite field of characteristic $2$.

\begin{theorem}\label{thm:multilin}
Let $P\in\FF[x_1,\dots,x_m]$ be a multilinear polynomial where $\FF$ is a finite field of characteristic $2$. 
Then the three following statements are equivalent:
\begin{itemize}
\item[$(i)$] $P$ is representable;
\item[$(ii)$] For every tuple of squares $\ell^2\in\FF^m$, $P$ is
factorizable \emph{modulo} $\I(\ell^2)$;
\item[$(iii)$] There exists a tuple of squares $\ell^2\in\FF^m$ such that
$P$ is factorizable \emph{modulo} $\I(\ell^2)$.
\end{itemize}
\end{theorem}

\begin{proof}
The implication $(i)\implies(ii)$ is a special case of Theorem~\ref{thm:nc}, and $(ii)\implies(iii)$ is evident. Let us prove that $(iii)\implies(i)$.

Let $\ell^2\in\FF^m$ such that $\pi_{\ell^2}(P) = t_1\times\dotsb\times
t_k$, where each $t_i$ is a linear element of $\Ring(\ell^2)$.
For $i\leq k$, $\rho_{\ell^2}(t_i)$ is a linear polynomial, hence by
Proposition~\ref{prop:representable}, we know that $Q =
\rho_{\ell^2}(t_1)\times\dots\times\rho_{\ell^2}(t_k)$ has an SDR $M$. 
By Lemma~\ref{lem:multilin}, there exists an SDR $N$ of $Q$ such that
each variable appears at most once on the diagonal.
Hence, Lemma~\ref{nc:proj} and Lemma~\ref{nc:diag} ensure that
$\pi_{\ell^2}(P) = \pi_{\ell^2}(Q)$ has a gSDR $A$ such that each
$A_{i,i}$ is linear and each $A_{i,j}$ is constant for $i\neq j$.
Let $O$ be the matrix defined by $O_{i,j} = \rho_{\ell^2}(A_{i,j})$.
Since each variable appears once on the diagonal, $\det(O)$ is a
multilinear polynomial.
We have $\pi_{\ell^2}(P) = \pi_{\ell^2}(Q) = 
\det(A) = \pi_{\ell^2}(\det(O))$.
Since both $P$ and $\det(O)$ are multilinear, we have $P = \det(O)$, hence
$P$ is representable.
\end{proof}

If $\FF$ is infinite, a similar characterization can be obtained. To this end, 
the conclusion of Theorem~\ref{thm:nc} can be reinforced as follows: If $P$ is
representable, then there exist linear polynomials $L_1$, \dots, $L_k$ \emph{whose
coefficients are quadratic residues in $\FF$} such that 
$\pi_{\ell^2}(P)=\pi_{\ell^2}(L_1\dotsm L_k)$. One can check that the proof of
Theorem~\ref{thm:nc} actually is a proof of this stronger statement. The converse
is proved using Proposition~\ref{prop:representable}.

\subsection{Towards a full characterization}\label{sec:full} 

Theorem~\ref{thm:nc} is valid for any polynomial. Thus we have a necessary
condition for all polynomials. The characterization for multilinear
polynomials relies on the fact that $\rho(\pi(P))=P$ in this case. If we
are working with a non-multilinear polynomial $P$, the projection of $P$
\emph{modulo} some ideal $\I(\ell)$ can dramatically change the structure
of the polynomial. In particular, if we have a polynomial $P=x_1^2\times
Q$ for some multilinear polynomial $Q$, then $\Mult_{(1,\dots)} P= Q$ but
$\Mult_{(0,\dots)} P=0$. Thus, it is certainly impossible to go back from
the projection \emph{modulo} $\I(0,\dots)$ to $P$. To come up with this
issue, we look at some new specific ideal for the projection.
In this section, the field $\FF$ is supposed to be finite. With the
same arguments as for Theorem~\ref{thm:multilin}, the results of this
section can be extended to any field of characteristic $2$.

Let $\FF(\xi_1,\dots,\xi_m)$ be the field of fractions in $m$
indeterminates over $\FF$, and 
$\I(\xi^2)=\langle x_1^2+\xi_1^2,\dots,x_m^2+\xi_m^2\rangle$. For
$P\in\FF[x_1,\dots,x_m]$, we can consider the multilinear polynomial
$\Mult_{\xi^2} P = \rho_{\xi^2}(\pi_{\xi^2}(P))$
and apply Theorem~\ref{thm:multilin} about multilinear polynomials. In
particular, $\Mult_{\xi^2} P $ is representable if, and only if, it is
factorizable. 

The problem we face is that our constructions use inverse of elements in
the base field. This means that we have an equivalence between
factorization and SDR for multilinear polynomials in
$\FF(\xi_1,\dotsc,\xi_m)[x_1,\dotsc,x_m]$ but the factorization or the SDR
we build can use rational fractions in the $\xi_i$'s. To partly avoid this
problem, we have to restrict the ideals we are working with to ideals of
the form $\I(\ell^2)$ for $\ell\in\FF^m$. Unfortunately, it is not
sufficient. Nevertheless, we are able to prove some partial results.

\begin{lemma}\label{lem:xi}
Let $P\in\FF[x_1,\dots,x_m]$. Then $P$ is representable if, and only if,
$\Mult_{\xi^2} P$ has an SDR with non-diagonal entries in $\FF[\xi_1,\dots,\xi_m]$.
\end{lemma}

\begin{proof} Let us remark at first that $\Mult_{\xi^2}$ is a bijection
from $\FF[x_1,\dotsc,x_m]$ to the set of multilinear polynomials with
coefficients in $\FF[\xi_1,\dotsc,\xi_m]$. Indeed, its inverse
$\Mult_{\xi^2}^{-1}$ simply consists in mapping each $\xi_i$ to $x_i$. 

Using Lemma~\ref{lem:multilin},
we can transform any (g)SDR to an SDR such that each variable appears
exactly once on the diagonal. 

Let $M$ be an SDR of $P$. We can apply the procedure $\Clean$ to $M$ (with
respect to the tuple $(\xi_1^2,\dotsc,\xi_m^2)$). This yields an SDR of
$\Mult_{\xi^2} P$ as proved by Lemma~\ref{nc:proj}. Conversely, if we have
an SDR $M'$ of $\Mult_{\xi^2} P$, we can replace each $\xi_i$ by $x_i$ to
get an SDR of $P$. This corresponds to applying $\Mult_{\xi^2}^{-1}$ to
each entry.  As this function is compatible with the addition and
multiplication, the new matrix $M$ we obtain satisfy $\det(M)=
\Mult_{\xi^2}^{-1}(\det(M'))=P$.
\end{proof}

\begin{theorem}
Let $P\in\FF[x_1,\dots,x_m]$.
\begin{itemize}
\item If $P$ is representable, then for every tuple of squares
$\ell^2\in\FF^m$, 
$\Mult_{\xi^2} P$ can be factorized as a product of linear polynomials
\emph{modulo} $\I(\ell^2)$, and the linear polynomials have coefficients
in $\FF(\xi_1,\dotsc,\xi_m)$. 
\item If $\Mult_{\xi^2} P$ can be factorized as a product of linear
polynomials \emph{modulo} $\I(\ell^2)$ for some tuple of squares $\ell^2$,
and if the linear polynomials have coefficients in
$\FF[\xi_1,\dots,\xi_m]$, then $P$ is representable.
\end{itemize}
\end{theorem}

To obtain a full characterization, we would need to prove that in the first statement, we can obtain linear factors with coefficients in $\FF[\xi_1,\dotsc,\xi_m]$. 

\begin{proof}

The first statement only consists in applying Theorem~\ref{thm:nc} to $P\in\FF(\xi_1,\dotsc,\xi_m)$.

For the second statement, suppose $\Mult_{\xi^2} P\equiv
L_1\times\dotsb\times L_k\mod\I(\ell^2)$ for some $\ell^2$, and each $L_i$
has coefficients in $\FF[\xi_1,\dotsc,\xi_m]$. Using
Theorem~\ref{thm:multilin}, we can build a matrix representing
$\Mult_{\xi^2} P$. Since the $L_j$'s have as coefficients some polynomials
in the $\xi_i$'s, and since the transformations of
Lemma~\ref{lem:multilin} used in the proof of Theorem~\ref{thm:multilin}
use no inverse of any of the coefficients, we get an SDR of $\Mult_{\xi^2}
P$ the non-diagonal entries of which are polynomials in the $\xi_i$'s.
Using Lemma~\ref{lem:xi}, we conclude that $P$ is representable.
\end{proof}

\section{Factorization} \label{sec:factor} 

Section~\ref{sec:multilin} gives a characterization of representable multilinear polynomials in terms of the factorization of the polynomials into linear polynomials \emph{modulo} an ideal $\I(\ell)$. We give in this section an algorithm to decide this problem. Its running time is polynomial in the number of monomials of the polynomial. In this section, $\FF$ is a finite field of characteristic $2$.

\subsection{Preliminary results} 

In the previous section, we worked with elements of the quotient ring $\Ring(\ell)$ for some tuple $\ell$. The algorithms presented in this section deal with multilinear polynomials $P\in\FF[x_1,\dotsc,x_m]$.
Theorem~\ref{thm:multilin} is the basic tool. Since $P=\Mult_\ell(P)$ for any $\ell$, it can be reformulated as follows:
A multilinear polynomial is representable if, and only if, for every tuple of squares $\ell^2$, there exist linear polynomials $L_1$, \dots, $L_k$ such that $P=\Mult_{\ell^2}(L_1\dotsm L_k)$. It is equivalent to say that $\pi_{\ell^2}(P)=\pi_{\ell^2}(L_1\dotsm L_k)$. Moreover, as seen before this existence does not depend on the tuple $\ell^2$. This motivates the following definition.

\begin{definition}
We say that a multilinear polynomial $P$ is \emph{factorizable} if there exist a tuple of squares $\ell^2$ 
and 
linear polynomials $L_1$, \dots, $L_k$ such that
\[P= \Mult_{\ell^2}(L_1\times\dotsb\times L_k).\]
\end{definition}

The algorithm heavily relies on the fact that the possibility to factorize a polynomial \emph{modulo} $\I(\ell^2)$ does not depend on $\ell^2$. Actually two tuples are used, $\bar 0=(0,\dotsc,0)$ and $\bar 1=(1,\dotsc,1)$. To simplify the notations, these tuples are respectively denoted by $0$ and $1$. Moreover $\pi_0$, $\rho_0$ and $\Mult_0$ on the one hand, and $\pi_1$, $\rho_1$ and $\Mult_1$ on the other hand, are the functions defined in Section~\ref{sec:quotientRings}. In the same way, let
\[\I_0=\I(\bar 0)=\langle x_1^2,\dots,x_m^2\rangle \quad\text{and}\quad \I_1=\I(\bar 1)=\langle x_1^2+1,\dots,x_m^2+1\rangle\]
and $\Ring_0$ and $\Ring_1$ be defined by analogy.

We shall sometimes write that $P$ is \emph{factorizable modulo $\I$} for $\I=\I_0$ or $\I_1$ instead of simply factorizable to emphasize the fact that we are working specifically with the ideal $\I$. 
Let $P$ be a multilinear polynomial. We define its \emph{linear part} $\lin(P)$ as the sum of all its monomials of degree at most $1$. For instance $\lin(xyz+xy+x+z+1)=x+z+1$. Furthermore, we write $\partial P/\partial x_i$ for the partial derivative of $P$ with respect to the variable $x_i$. For a multilinear polynomial, this equals the quotient in the euclidean division of $P$ by $x_i$.

To show how to test the factorizability of a multilinear polynomial, we
proceed in two steps. We first show how to test the factorizability of a
polynomial $P$ whose monomial of lowest degree has degree exactly $1$ (we
say that $P$ has \emph{valuation} $1$). 
To this end, we show  that $P$ is factorizable if, and only if, 
$P=\Mult_0(\lin(P)\times\frac{1}{\alpha_i}\frac{\partial P}{\partial x_i})$
where $\alpha_ix_i$ is a nonzero monomial of $\lin(P)$ (Lemmas~\ref{lem:oneFactor} 
and~\ref{lem:factor}).
The second step proves that given any multilinear
polynomial $P$, we can compute a polynomial $Q$ of the same degree 
whose valuation is $1$
such that $P$ is factorizable if, and only if, $Q$ is also. There are two
cases, covered by Lemmas~\ref{lem:fullPoly} and \ref{lem:LinPart}. This
will allow us to describe an algorithm using alternatively those two steps
to test factorizability.

\begin{lemma}\label{lem:oneFactor}
Let $P$ be a multilinear polynomial 
of valuation $1$. 
If there exists some linear polynomials $L_1$, \dots, $L_k$ such that
\[P= \Mult_0(L_1\times\dotsb\times L_k),\]
then there exist an index $j$ and a constant $\alpha$ such that $\lin(P)=\alpha L_j$.
\end{lemma}

\begin{proof}
Suppose that $P=\Mult_0(L_1\dotsm L_k)$ and let $Q=L_1\dotsm L_k$. In particular, $Q(0)=0$ and $\lin(P)=\lin(Q)$. Thus there exists $j$ such that $L_j(0)=0$. In other words, $L_j$ is a sum of degree-$1$ monomials. The linear part of $P$ being nonzero, the polynomial $Q/L_j$ has a constant coefficient $\alpha\in\FF$. A degree-$1$ monomial of $Q$ is the product of a monomial of $L_j$ by $\alpha$. This means that $\lin(P)=\lin(Q)=\alpha L_j$. 
\end{proof}

We now prove that we can efficiently test if some linear polynomial $L$ can appear in the factorization of $P$ \emph{modulo} $\I_0$. 

\begin{lemma}\label{lem:factor}
Let $P$ be a multilinear polynomial and $L$ be a linear polynomial with no constant coefficient, having a nonzero monomial $\alpha_i x_i$. 
If there exists a multilinear polynomial $Q$ such that $P=\Mult_0(L\times Q)$, then
\[P=\Mult_0\left(L\times\frac{1}{\alpha_i}\frac{\partial P}{\partial x_i}\right).\]
\end{lemma}

\begin{proof}
Suppose that $P=\Mult_0(L\times Q)$. This means that there exist polynomials $p_1$, \dots, $p_m$ such that
\[P=L\times Q+x_1^2p_1+\dotsb+x_m^2p_m.\]
Moreover, $\partial(x_j^2p_j)/\partial x_i=x_j^2\partial p_j/\partial x_i$ for all $j$. For $j=i$, this comes from the fact that $\partial x_i^2/\partial x_i=2x_i=0$. Since $L$ contains the monomial $\alpha_i x_i$, $\partial L/\partial x_i=\alpha_i$. Therefore,
\[\frac{\partial P}{\partial x_i}=\alpha_i Q + L\frac{\partial Q}{\partial x_i}+\frac{\partial p_1}{\partial x_i} x_1^2+\dotsb+\frac{\partial p_m}{\partial x_i} x_m^2.\]
Since $\alpha_i\neq 0$, the previous equality can be multiplied by $L/\alpha_i$ to obtain
\[L\times\frac{1}{\alpha_i}\frac{\partial P}{\partial x_i}=L\times Q + \frac{L^2}{\alpha_i}\frac{\partial Q}{\partial x_i}+\frac{L}{\alpha_i}\left(\frac{\partial p_1}{\partial x_i} x_1^2+\dotsb+\frac{\partial p_m}{\partial x_i} x_m^2\right).\]
Since $L^2$ is a sum of squares,
\[\Mult_0\left(L\times\frac{1}{\alpha_i}\frac{\partial P}{\partial x_i}\right)=\Mult_0(L\times Q)=P.\]
\end{proof}

In the second step we prove that given any multilinear polynomial $P$, we
can find a polynomial $Q$ of valuation $1$ such that $P$ is factorizable
if, and only if, $Q$ is also. There are two distinct cases. At first, we focus 
on \emph{full} polynomials, that is without zero coefficient. An $m$-variate
multilinear polynomial is full if it has $2^m$ nonzero monomials. In particular,
if every coefficient equals $1$, then the polynomial can be factorized as
$\prod_i(1+x_i)$. In the general case, such a factorization does not necessarily 
exist.

In the following lemma, given a full polynomial, a new polynomial is produced 
that is either not full, or has less variables. In any case, the number of monomials
decreases.

\begin{lemma}\label{lem:fullPoly}
Let $P$ be a multilinear polynomial in $m$ variables with exactly $2^m$ monomials.
Then there exists a linear polynomial $L$ such that $Q\defeq \Mult_0(P\times L)$ is nonzero
and has less than $2^m$ monomials. 
Moreover, $P$ is factorizable if, and only if, $Q$ is factorizable. 
\end{lemma}

\begin{proof}
Let $x_i$ be any variable of $P$ and $L=p_i x_i+p_0$ where $p_i$ is the coefficient of $x_i$ in $P$ and $p_0$ is its constant coefficient. 
Then the constant coefficient of $Q$ is $p_0^2$ and thus $Q$ is nonzero, and the coefficient of $x_i$ in $Q$ is $p_0 p_i+p_i p_0=0$. 
Thus $Q$ has less monomials than $P$. 

By definition, $Q$ is factorizable if $P$ is also. Moreover, $\Mult_0(L\times Q)=\Mult_0(L^2\times P)=\Mult_0(p_0^2P)=p_0^2P$ since $P$ is multilinear. Thus if $Q$ is factorizable, then so is $P$.
\end{proof}

It remains to deal with the case where $P$ does not possess all the possible monomials. In this lemma, we consider the ideal $\I_1$ instead of $\I_0$ as before.

\begin{lemma}\label{lem:LinPart}
Let $P$ be a multilinear polynomial over $m$ variables with at most $(2^m-1)$ monomials. 
Then there exists a primitive monomial $x^\alpha$ such that $Q\defeq \Mult_1(x^\alpha\times P)$ has valuation $1$. 
Moreover, $P$ is factorizable if, and only if, $Q$ is factorizable.
\end{lemma}

\begin{proof}
If $P$ already has valuation $1$, we can take 
$\alpha=(0,\dotsc,0)$.

If $P$ has valuation greater than $1$, 
let $x^{\beta}$ be a nonzero primitive monomial of $P$ of minimal degree. Let $i$ be some index such that $\beta_i\neq0$ 
and define $x^\alpha=x^\beta/x_i$. Then $\Mult_1(x^\alpha x^\beta)=x_i$. Moreover, $\Mult_1(x^\alpha x^\gamma)=1$ if, and only if, $\alpha=\gamma$. Since $\deg x^\alpha<\deg x^\beta$, the coefficient of $x^\alpha$ in $P$ is zero. Thus $x^\alpha$ satisfies the lemma.

If $P$ has valuation $0$, 
let $x^\alpha$ be a primitive monomial of minimal degree whose coefficient in $P$ is zero. Such a monomial exists since $P$ has at most $(2^m-1)$ monomials. Then $\Mult_1(x^\alpha\times P)$ has no constant coefficient. 
Furthermore, by the minimality of $x^\alpha$, every monomial of smaller degree has a nonzero coefficient in $P$. This is in particular the case of the monomials $x^\alpha/x_i$ for every variable $x_i$ that divides $x^\alpha$. Since $\Mult_1(x^\alpha (x^\alpha/x_i)) =x_i$,  $x^\alpha$ satisfies the lemma. 

To finish the proof, we remark that $\Mult_1(x^\alpha Q)=\Mult_1((x^\alpha)^2 P)=P$. Therefore, $P$ is factorizable if, and only if, $Q$ is factorizable. 
\end{proof}

For the proof of the next corollary, one needs to remark a simple fact. If a multilinear polynomial $P$ is representable, then $\partial P/\partial x_i$ is also representable for any variable $x_i$. Indeed, suppose that $M$ is an SDR of $P$ with each variable appearing at most once on the diagonal. Suppose that $M_{j,j}=x_i$. Then the determinant of the matrix obtained by removing the row and column $j$ from $M$ equals $\partial P/\partial x_i$.

\begin{corollary}\label{cor:baseField}
Let $P\in\FF[x_1,\dotsc,x_m]$ be a multilinear polynomial and $\mathbb G/\FF$ a field extension. If $P$ has an SDR with entries in $\mathbb G\cup\{x_1,\dotsc,x_m\}$, then it has an SDR with entries in $\FF\cup\{x_1,\dotsc,x_m\}$.
\end{corollary}

\begin{proof}
Let us first consider $\mathbb G$ as the base field. By Theorem~\ref{thm:multilin}, $P$ is representable if, and only if, it is factorizable. Using Lemmas~\ref{lem:fullPoly} and~\ref{lem:LinPart}, one can suppose that $P$ has valuation $1$. Moreover, using Lemmas~\ref{lem:oneFactor} and~\ref{lem:factor} and the remark before the corollary, one can deduce that $P$ is representable if, and only if, $P=\Mult_0(\lin(P)\times\frac{1}{\alpha_i}\frac{\partial P}{\partial x_i})$ and $\partial P/\partial x_i$ is representable. 

Now, $\partial P/\partial x_i$ is a polynomial with coefficients in $\FF$, which has an SDR with entries in $\mathbb G\cup\{x_1,\dotsc,x_m\}$. Moreover, it has one variable less than $P$. Therefore, we can prove the corollary by induction on the number of variables.
\end{proof}

\subsection{An algorithm for factorizability} 

The previous lemmas yield a polynomial time algorithm to decide whether some polynomial $P$ is factorizable.
We first give an algorithm $\Prep$ (Algorithm~\ref{algo:prep}) corresponding to Lemmas~\ref{lem:fullPoly} and~\ref{lem:LinPart}. 

\begin{algorithm2e}[htbp]
\DontPrintSemicolon
\caption{$\Prep(P)$}
\label{algo:prep}
\KwIn{A multilinear polynomial $P$}
\KwOut{A multilinear polynomial $Q$ of valuation $1$ or linear}
\BlankLine
\lIf{$P$ is linear}{\Return $P$}\;
\tcp*[l]{Lemma~\ref{lem:fullPoly}:}
\uElseIf{$P$ is full}{
    $x_i\gets$ some variable of $P$\;
    $p_i\gets$ coefficient of $x_i$ in $P$\;
    $p_0\gets$ constant coefficient of $P$\;
    $P\gets\Mult_0(P\times(p_ix_i+p_0))$\; \label{line:full}
    \Return $\Prep(P)$\;  
    }
\tcp*[l]{Lemma~\ref{lem:LinPart}:}
\uElseIf{$P$ has valuation $0$}{
    $x^\alpha\gets$ minimal monomial with a zero coefficient in $P$\;
    \Return $\Mult_1(x^\alpha P)$\;                 \label{line:val0}
    }
\uElseIf{$P$ has valuation $>1$}{
    $x^\alpha\gets$ minimal monomial of $P$, divided by one of its variables\;
    \Return $\Mult_1(x^\alpha P)$\;                 \label{line:val2}
}
\lElse{\Return $P$}\;
\end{algorithm2e}

\begin{lemma}\label{lem:prep}
Let $P\in\FF[x_1,\dotsc,x_m]$ be a multilinear polynomial. Then $Q=\Prep(P)$ is either linear, or has valuation $1$. Moreover, $P$ is factorizable if, and only if, $Q$ is also.

The algorithm runs in time polynomial in the number of variables and the number of monomials of $P$.
\end{lemma}

\begin{proof}
The correctness is ensured by Lemmas~\ref{lem:fullPoly} and~\ref{lem:LinPart}. We only have to prove its termination and the running time estimate. There is a recursive call when $P$ is full. From Lemma~\ref{lem:fullPoly}, $\Mult_0(P\times(p_0x_i+p_i))$ then has at most $(2^m-1)$ monomials. Either this new polynomial is full, but the number of variables decreased, or the condition ``$f$ is full'' is not satisfied anymore and there is no new recursive call. Therefore, the number of recursive calls is bounded by the number of variables. This proves both the termination and the complexity analysis.
\end{proof}

Let us now describe the algorithm $\IsFactor$ (Algorithm~\ref{algo:factor}) corresponding to Lemmas~\ref{lem:oneFactor} and~\ref{lem:factor}.

\begin{algorithm2e}[htbp]
\DontPrintSemicolon
\caption{$\IsFactor(P)$}
\label{algo:factor}
\KwIn{A multilinear polynomial $P$}  
\KwOut{Is $P$ factorizable?}
\BlankLine
$P\gets\Prep(P)$\; 
\lIf{$P$ is linear}{\Return\textbf{True}}\;     \label{line:linear}
\BlankLine
\Else{  
    \tcp*[l]{Lemmas~\ref{lem:oneFactor} 
        \& \ref{lem:factor}:}
    $\alpha_ix_i\gets$ some nonzero monomial of $\lin(P)$\;          \label{line:recB} 
    $P_0\gets \frac{\partial P}{\partial x_i}$\;    \label{line:quotient}
    
    \eIf{$P=\Mult_0(\frac{1}{\alpha_i}\lin(P)\times P_0)$}{ \label{line:div}
        \Return $\IsFactor(P_0)$\;              \label{line:call}
    }{
        \Return \textbf{False}                  \label{line:recE}
    }
}
\end{algorithm2e}

\begin{theorem}\label{thm:factor}
Let $P\in\FF[x_1,\dotsc,x_m]$ be a multilinear polynomial. Then the algorithm $\IsFactor(P)$ answers \textbf{True} if, and only if, $P$ is factorizable.

The algorithm runs in time polynomial in the number $m$ of variables and the number of monomials of $P$.
\end{theorem}

\begin{proof}
The correctness follows from Lemmas~\ref{lem:oneFactor} and~\ref{lem:factor}. The termination is ensured by the fact that $\partial P/\partial x_i$ has less variables than $P$. This bounds the number of iterations by $m$.
\end{proof}

\subsection{An algorithm for the representation} 

Algorithm~\ref{algo:factor} only tells us if a polynomial is factorizable, but does not give us a factorization. The reason for this is that we change several times the ideal we are working with. Nevertheless, we proved in Section~\ref{sec:multilin} that given the factorization of a multilinear polynomial \emph{modulo} some ideal $\I$, we can find a symmetric matrix representing the polynomial. We use this in the following to show how to modify Algorithm~\ref{algo:factor} in order to get a Symmetric Determinantal Representation. Using the results of Section~\ref{sec:nec-cdt}, we are then able to factorize any factorizable multilinear polynomial \emph{modulo} any ideal $\I$.

\begin{lemma}\label{lem:merge}
Given two SDRs $M_P$ and $M_Q$ of two multilinear polynomials $P$ and $Q$ respectively, one can build an SDR $\Merge_b(M_P,M_Q)$ of $\Mult_b(P\times Q)$ ($b\in\{0,1\}$) in time polynomial in the dimensions of $M_P$ and $M_Q$.
\end{lemma}

\begin{proof}
The algorithm is based on Lemma~\ref{lem:multilin}. Let $N$ be the block-diagonal matrix made of $M_P$ and $M_Q$. Clearly, this matrix represents $P\times Q$. Using Lemma~\ref{lem:multilin}, one can build a matrix $N'$ such that each variable appears at most once on the diagonal and such that $\det(N')=P\times Q$. Then $\Merge_b(M_P,M_Q)=\Mult_b(N')$ represents $\Mult_b(P\times Q)$.
\end{proof}

\begin{theorem}
There is an algorithm $\SymDet$ that given as input a multilinear polynomial $P\in\FF[x_1,\dots,x_m]$ 
returns an SDR of $P$ if one exists. 
This algorithm runs in time polynomial in $m$ and the number of monomials.
\end{theorem}

\begin{proof}
The algorithm $\SymDet$ is made of two steps. The first one is a modification of the algorithm $\IsFactor$ such that it returns a list of factors instead of \textbf{True} when $P$ is factorizable. The second one is the construction of an SDR of $P$ from this list of factors, using the algorithm $\Merge$ of Lemma~\ref{lem:merge}. 

In algorithms $\Prep$ and $\IsFactor$, to test if $P$ is factorizable, it is written as $P=\Mult_b(L\times Q)$ where $L$ is either linear or a monomial and $b \in\{0,1\}$, and then one tests whether $Q$ is factorizable. These algorithms are modified to retain the couples $(L,b)$ each time such an operation is performed. More precisely, we add a global variable $\L$ containing a list of pairs of the form $(L,b)$. Let us now describe how $\Prep$ and $\IsFactor$ modify this variable.

In $\Prep$, Line~\ref{line:full}, the couple $((p_0x_i+p_i)/p_i^2,0)$ is added to $\L$. Indeed, a recursive call is performed with the polynomial $Q=\Mult_0(P\times(p_0x_i+p_i))$. But $\Mult_0(Q\times(p_0x_i+p_i)/p_i^2)$ equals $\Mult_0(P\times(p_0^2x_i^2+p_i)^2/p_i^2)=P$. In the same way, the couple $(x^\alpha,1)$ is added to $\L$ at Lines~\ref{line:val0} and~\ref{line:val2}. To finish, the couple $(\lin(f)/\alpha_i,0)$ is added to $\L$ at Line~\ref{line:call} of $\IsFactor$.

When $\IsFactor$ answers \textbf{True}, then $P$ is linear. Instead of this answer, the new algorithm adds the couple $(P,0)$ to $\L$ (the bit $0$ is arbitrary and unused) and returns $\L$. At this stage, we have a list $\L$ of couples $(L_1,b_1)$, \dots, $(L_k,b_k)$. Let $P_k=L_k$ and for $i$ from $(k-1)$ down to $1$, $P_i=\Mult_{b_i}(L_i\times P_{i+1})$. From the construction of $\L$, $P=P_1$. An SDR for $P$ is built as follows: For all $i$, an SDR $N_i$ of $L_i$ is built using Proposition~\ref{prop:representable}. Then, let $M_k=N_k$ and for $i$ from $(k-1)$ down to $1$, let $M_i=\Merge_{b_i}(N_i,M_{i+1})$. If $\det(M_{i+1})=P_{i+1}$ and $\det(N_i)=L_i$, Lemma~\ref{lem:merge} shows that $\det(M_i)=P_i$. To conclude, the algorithm returns $M=M_1$. 

The running time of the algorithm is controlled by the running times of $\Prep$, $\IsFactor$, and $\Merge$.
\end{proof}

\section{A characteristic-free result: Alternating Determinantal
Representations}\label{sec:alternating} 
Symmetric matrices correspond to symmetric bilinear forms. We saw that in
this context, there is a big difference depending on whether the
characteristic of the underlying field is $2$ or not.
As explained to us by Mathieu Florence \cite{Flo2012}, the related notion
of alternating forms is known to be \emph{characteristic-free}.
An \emph{alternating form} is a bilinear form $\varphi: V\times V \to \FF$
such that $\varphi(v,v) = 0$ for any $v$ in the vector space
$V$.
The matrix associated to alternating forms are the anti-symmetric matrix
with zero diagonal entries. 
Hence, we should expect an homogeneous result concerning Alternating
Determinantal Representations. It turns out to be the case:
\begin{theorem}\label{thm:alternating}
Let $\FF$ be some field and $P\in\FF[x_1,\dots,x_m]$ be a polynomial.
Then, $P$ can be written as the determinant of an alternating matrix with
entries in $\FF\cup\{x_1,\dots,x_m\}$ if, and only if, $P$ is a square.
\end{theorem}
\begin{proof}
Let $P$ be the determinant of an alternating matrix $M$ with entries in
$\FF\cup\{x_1,\dots,x_m\}$.  If we consider $M$ as a matrix over the
commutative ring $\FF[x_1,\dots,x_m]$, we see that $P=\det(M)$ is the
square of the Pfaffian of $M$ which is an element of $\FF[x_1,\dots,x_m]$
(\cite{Lang2002} XV, \textsection 9, page 588).
\newline
Conversely, let $P=Q^2$ be the square of an element of
$\FF[x_1,\dots,x_m]$. As proved in \cite{Val79-complete}, there exists a
square matrix $N$ with entries in $\FF\cup\{x_1,\dots,x_m\}$ such that
$\det(N) = Q$. The matrix 
\[M=\begin{pmatrix}0 & N \\ -N^T & 0 \end{pmatrix}\]
is alternate and satisfies $\det(M) = (\det(N))^2 = Q^2 = P$.
\end{proof}

\section{Concluding remarks}\label{sec:conclusion} 

We proved in this paper that in characteristic $2$, some polynomials do not admit any SDR. In the case of multilinear polynomials, we gave a complete characterization as well as algorithms to deal with these representations. We discovered some tight relations between the ability to find an SDR and to factorize the polynomial \emph{modulo} some square polynomials. Thus we showed that the factorization in these quotient rings can be performed in polynomial time.

The main remaining open question is of course to get a full characterization of representable polynomials. 

\begin{conjecture}
A polynomial $P\in\FF[x_1,\dotsc,x_m]$ is representable if, and only if, for some (equivalently any) tuple of squares $\ell\in\FF^m$, $\Mult_\xi P$ is factorizable \emph{modulo} $\I(\ell)$ into linear polynomials $L_1,\dotsc,L_k\in\FF[\xi_1,\dots,\xi_m][x_1,\dots,x_m]$. 
\end{conjecture}

An example of a problematic polynomial is $x_1^2+x_1x_2+x_1x_3+x_2x_3$. Indeed, it can be factorized as $(x_1+x_2)(x_1+x_3)$. But once the projection \emph{modulo} $\I_\xi$ is made, it is not so clear anymore how it can be factorized. One idea could be first to factorize the polynomial and then apply our results to each factor. Yet it is not clear whether this strategy can work.

If the conjecture can be proved, or if some other characterization of the same kind can be found, it would also remain to see if the algorithms designed for multilinear polynomials can be extended to the general case. Once again, the main difficulty is to deal with the fact that our algorithms use inverse of elements.

Even for multilinear polynomials, some questions remain. For instance, it could be interesting to make a quantitative study to know how many polynomials have SDRs. For instance, all polynomials in $2$ variables are representable, and it seems that the proportion decreases as the number of variables increases.

\subsection*{Acknowledgments} B.G. thanks Erich L. Kaltofen, Pascal Koiran, Natacha Portier, Yann Strozecki and Sébastien Tavenas for fruitful discussions on the subject of this paper.


\begin{thebibliography}{10}

\bibitem{Alb38}
{\sc Albert, A.}
\newblock {Symmetric and alternate matrices in an arbitrary field, I}.
\newblock {\em Trans. Amer. Math. Soc. 43}, 3 (1938), 386--436.
\newblock
  \textsc{doi}:\href{http://dx.doi.org/10.1090/S0002-9947-1938-1501952-6}{10.1090/S0002-9947-1938-1501952-6}.

\bibitem{Bea00}
{\sc Beauville, A.}
\newblock {Determinantal hypersurfaces}.
\newblock {\em Michigan Math. J 48\/} (2000), 39--64.
\newblock
  \textsc{doi}:\href{http://dx.doi.org/10.1307/mmj/1030132707}{10.1307/mmj/1030132707}.

\bibitem{Bra10}
{\sc Br{\"a}nd{\'e}n, P.}
\newblock Obstructions to determinantal representability.
\newblock {\em Adv. Math. 226}, 2 (2011), 1202--1212.
\newblock arXiv:\href{http://arxiv.org/abs/1004.1382}{1004.1382},
  \textsc{doi}:\href{http://dx.doi.org/10.1016/j.aim.2010.08.003}{10.1016/j.aim.2010.08.003}.

\bibitem{Bra12}
{\sc Br{\"a}nd{\'e}n, P.}
\newblock {Hyperbolicity cones of elementary symmetric polynomials are
  spectrahedral}.
\newblock arXiv:\href{http://arxiv.org/abs/1204.2997}{1204.2997}, 2012.

\bibitem{Dic21}
{\sc Dickson, L.}
\newblock {Determination of all general homogeneous polynomials expressible as
  determinants with linear elements}.
\newblock {\em Trans. Amer. Math. Soc 22\/} (1921), 167--179.
\newblock
  \textsc{doi}:\href{http://dx.doi.org/10.1090/S0002-9947-1921-1501168-0}{10.1090/S0002-9947-1921-1501168-0}.

\bibitem{Diestel}
{\sc Diestel, R.}
\newblock {\em Graph Theory}, 3rd~ed.
\newblock Graduate Texts in Mathematics. Springer, 2006.

\bibitem{Dix02}
{\sc Dixon, A.}
\newblock {Note on the reduction of a ternary quantic to a symmetrical
  determinant}.
\newblock In {\em Proc. Cambridge Phil. Soc.\/} (1902), vol.~2, pp.~350--351.

\bibitem{Flo2012}
{\sc Florence, M.}
\newblock private communication, 2012.

\bibitem{GKKP11}
{\sc Grenet, B., Kaltofen, E.~L., Koiran, P., and Portier, N.}
\newblock {S}ymmetric {D}eterminantal {R}epresentation of {F}ormulas and
  {W}eakly {S}kew {C}ircuits.
\newblock In {\em Randomization, Relaxation, and Complexity in Polynomial
  Equation Solving}, L.~Gurvits, P.~P{\'e}bay, J.~M. Rojas, and D.~C. Thompson,
  Eds., no.~556 in Contemp. Math. Amer. Math. Soc., Providence, RI, 2011,
  pp.~61--96.
\newblock arXiv:\href{http://arxiv.org/abs/1007.3804}{1007.3804},
  \textsc{doi}:\href{http://dx.doi.org/10.1090/conm/556}{10.1090/conm/556},
  extended abstract in STACS'11.

\bibitem{HV07}
{\sc Helton, J., and Vinnikov, V.}
\newblock {Linear matrix inequality representation of sets}.
\newblock {\em Commun. Pur. Appl. Math. 60}, 5 (2007), 654--674.
\newblock arXiv:\href{http://arxiv.org/abs/math/0306180}{math/0306180},
  \textsc{doi}:\href{http://dx.doi.org/10.1002/cpa.20155}{10.1002/cpa.20155}.

\bibitem{HMV06}
{\sc Helton, J.~W., McCullough, S.~A., and Vinnikov, V.}
\newblock Noncommutative convexity arises from linear matrix inequalities.
\newblock {\em J. Funct. Anal. 240}, 1 (2006), 105--191.
\newblock
  \textsc{doi}:\href{http://dx.doi.org/10.1016/j.jfa.2006.03.018}{10.1016/j.jfa.2006.03.018}.

\bibitem{Lang2002}
{\sc Lang, S.}
\newblock {\em Algebra}, third~ed., vol.~211 of {\em Graduate Texts in
  Mathematics}.
\newblock Springer-Verlag, New York, 2002.
\newblock
  \textsc{doi}:\href{http://dx.doi.org/10.1007/978-1-4613-0041-0}{10.1007/978-1-4613-0041-0}.

\bibitem{NPT11}
{\sc Netzer, T., Plaumann, D., and Thom, A.}
\newblock {Determinantal Representations and the Hermite Matrix}.
\newblock {\em Mich. Math. J.\/} (2011).
\newblock arXiv:\href{http://arxiv.org/abs/1108.4380}{1108.4380}, to appear.

\bibitem{NT12}
{\sc Netzer, T., and Thom, A.}
\newblock Polynomials with and without determinantal representations.
\newblock {\em Linear Algebra Appl. 437}, 7 (2012), 1579--1595.
\newblock arXiv:\href{http://arxiv.org/abs/1008.1931}{1008.1931},
  \textsc{doi}:\href{http://dx.doi.org/10.1016/j.laa.2012.04.043}{10.1016/j.laa.2012.04.043}.

\bibitem{PSV11}
{\sc Plaumann, D., Sturmfels, B., and Vinzant, C.}
\newblock {Quartic curves and their bitangents}.
\newblock {\em J. Symb. Comput. 46}, 6 (2011), 712--733.
\newblock arXiv:\href{http://arxiv.org/abs/1008.4104}{1008.4104},
  \textsc{doi}:\href{http://dx.doi.org/10.1016/j.jsc.2011.01.007}{10.1016/j.jsc.2011.01.007}.

\bibitem{Qua12}
{\sc Quarez, R.}
\newblock Symmetric determinantal representation of polynomials.
\newblock {\em Linear Algebra Appl. 436}, 9 (2012), 3642--3660.
\newblock
  \textsc{hal}:\href{http://hal.archives-ouvertes.fr/hal-00275615}{hal-00275615},
  \textsc{doi}:\href{http://dx.doi.org/10.1016/j.laa.2012.01.004}{10.1016/j.laa.2012.01.004}.

\bibitem{Sch09}
{\sc Schweighofer, M.}
\newblock {Describing convex semialgebraic sets by linear matrix inequalities},
  2009.
\newblock Tutorial Session at ISSAC'09.
  \url{http://www.math.uni-konstanz.de/~schweigh/presentations/dcssblmi.pdf}.

\bibitem{Sage}
{\sc Stein, W., et~al.}
\newblock {\em {S}age {M}athematics {S}oftware ({V}ersion 4.5.3)}.
\newblock The Sage Development Team, 2010.
\newblock \url{http://www.sagemath.org}.

\bibitem{Val79-complete}
{\sc Valiant, L.~G.}
\newblock Completeness classes in algebra.
\newblock In {\em Proc. STOC'79\/} (1979), pp.~249--261.
\newblock
  \textsc{doi}:\href{http://dx.doi.org/10.1145/800135.804419}{10.1145/800135.804419}.

\bibitem{Wat85}
{\sc Waterhouse, W.}
\newblock {Symmetric determinants and Jordan norm similarities in
  characteristic 2}.
\newblock {\em Proc. Amer. Math. Soc. 93}, 4 (1985), 583--589.
\newblock
  \textsc{doi}:\href{http://dx.doi.org/10.1090/S0002-9939-1985-0776183-2}{10.1090/S0002-9939-1985-0776183-2}.

\end{thebibliography}
\end{document}